\documentclass[11pt,authoryear]{elsarticle}
\journal{Econometrics and Statistics}

\usepackage{fullpage}
\usepackage{amsmath, amsthm, amssymb}
\usepackage{bbm}
\usepackage{mathtools}
\usepackage{graphicx}
\usepackage{xcolor}
\usepackage{enumerate}
\usepackage{natbib}
\usepackage{subfig}
\usepackage{multirow}
\usepackage{float}
\usepackage{tikz}
\usepackage{longtable}
\usepackage{indentfirst}
\usepackage{setspace}
\usepackage{xurl}
\RequirePackage[colorlinks = true, linkcolor = blue, citecolor=blue,urlcolor=blue]{hyperref}

\newtheorem{lemma}{Lemma}
\newtheorem{cor}{Corollary}
\newtheorem{theorem}{Theorem}

\newtheorem{prop}{Proposition}

\DeclareMathOperator{\MSE}{MSE}
\DeclareMathOperator{\RCV}{RCV}
\DeclareMathOperator{\MAD}{MAD}

\begin{document}
\begin{frontmatter}
\title{Robust thin-plate splines for multivariate spatial smoothing}
\author{Ioannis Kalogridis}
\ead{ioannis.kalogridis@kuleuven.be}
\address{Department of Mathematics, KU Leuven}

\begin{abstract}
A novel family of multivariate robust smoothers based on the thin-plate (Sobolev) penalty that is particularly suitable for the analysis of spatial data is proposed. The proposed family of estimators can be expediently computed even in high dimensions, is invariant with respect to rigid transformations of the coordinate axes and can be shown to possess optimal theoretical properties under mild assumptions. The competitive performance of the proposed thin-plate spline estimators relative to their non-robust counterpart is illustrated in a simulation study and a real data example involving two-dimensional geographical data on ozone concentration.

\end{abstract}


\end{frontmatter}

\section{Introduction}

Consider the problem of estimating the regression function $f_0: \mathbbm{R}^d \to \mathbbm{R}$ from $n$ observations $(\mathbf{x}_i, Y_i) \in \mathbbm{R}^d \times \mathbbm{R},\ i = 1, \ldots, n,$ following the model
\begin{align}
Y_i = f_0(\mathbf{x}_i) + \epsilon_i, \quad (i=1, \ldots, n),
\end{align}
where the $\epsilon_i$ are random errors that are commonly assumed to be independent and identically distributed (i.i.d.) with mean zero and finite variance, but we will be able to considerably relax these assumptions over the course of this paper. Models of this general type arise naturally throughout the sciences, as very often empirical data cast doubt on parametric regression assumptions \citep[Chapter 5]{Wood:2017}.

A popular method of estimation of $f_0$ that is expounded by \citet{Wahba:1990} and \citet{Green:1994} involves restricting $f_0$ to the multivariate Sobolev space of functions of order $m$, $\mathcal{H}^{m}(\mathbb{R}^d)$, i.e., the space of all functions whose partial derivatives of total order $m$ for some $m \in \mathbbm{N}_{+}$ are in $\mathcal{L}^2(\mathbbm{R}^d)$. Mathematically, the space $\mathcal{H}^{m}(\mathbb{R}^d)$ is defined as
\begin{align*}
\mathcal{H}^{m}(\mathbb{R}^d) = \left\{ f:\mathbbm{R}^d \to \mathbbm{R}, \  \frac{\partial^m f(x_1, \ldots, x_d)}{\partial x_{1}^{m_1} \ldots \partial x_{d}^{m_d}} \ \text{exists for all} \ m_1 + \ldots + m_d = m \  \text{and} \  I_m^2(f) < \infty \right\},
\end{align*}
where the semi-norm $I_{m}:\mathcal{H}^{m}(\mathbb{R}^d) \to \mathbbm{R}_{+}$ is given by
\begin{align}
\label{eq:2}
I_{m}^2(f) = \sum_{m_1 + \ldots + m_d = m} \binom{m}{m_1, \ldots, m_d} \int_{\mathbbm{R}} \ldots \int_{\mathbbm{R}} \left( \frac{\partial^m f(x_1, \ldots, x_d)}{\partial x_{1}^{m_1} \ldots \partial x_{d}^{m_d}}  \right)^2 d x_1 \ldots d x_d.
\end{align}
Classical least squares thin-plate spline estimators are defined as minimizers of 
\begin{align*}
\frac{1}{2n}\sum_{i=1}^n \left(Y_i- f(\mathbf{x}_i)\right)^2 + \lambda I_m^2(f),
\end{align*}
over $\mathcal{H}^{m}(\mathbb{R}^d)$, for some $\lambda>0$ that governs the trade-off between smoothness and goodness of fit to the data. It is easy to see that for $d=1$ the penalty $I_{m}^2(f)$ reduces to $\int_{\mathbbm{R}} |f^{(m)}(x)|^2 dx$ and with standard arguments this may be simplified further to $\int_{\min{x}_i}^{\max{x}_i} |f^{(m)}(x)|^2 dx$, which gives rise to classical smoothing spline estimators, see, e.g., \citet{Wahba:1990}. Such simplifications are not valid for $d>1$, but the problem still admits an elegant solution provided only that $2m > d$. In fact, for $2m>d$, the solution to this variational problem may be written in terms of $n+M$ radial and polynomial basis functions with $M = \binom{m+d-1}{d}$, see \citet[p. 216]{Wood:2017} for more details. The resulting least squares thin-plate spline can be shown to converge at a fast rate under the aforementioned assumptions on the error term \citep{Cox:1984, G:2010} with the result that least squares thin-plate estimators combine computational efficiency with good theoretical properties.

An important drawback of least squares based estimators is their lack of resistance towards atypical observations. That is, these estimators tend to be overly attracted towards observations that do not follow the bulk of the data, which often leads to poor explanatory or predictive performance. In order to remedy this deficiency, this paper introduces and studies a family of generalized (M-type) thin-plate spline estimators defined as minimizers of
\begin{align}
\label{eq:3}
L_n(f) = \frac{1}{n}\sum_{i=1}^n \rho(Y_i - f(\mathbf{x}_i)) + \lambda I_m^2(f),
\end{align}
over $\mathcal{H}^{m}(\mathbb{R}^d)$, for some convex loss function $\rho:\mathbbm{R} \to \mathbbm{R}_{+}$. The square loss $\rho(x) = x^2/2$ is a typical example of such a loss function, but the generality of the above formulation also permits loss functions that increase less sharply as $|x| \to \infty$ thereby providing better protection against outliers. Notable examples include Huber, quantile and $L^q$ loss functions. As we discuss below, a  minimizer of \eqref{eq:3} in $\mathcal{H}^{m}(\mathbbm{R}^d)$ exists under fairly general conditions and, as a result, to identify this minimizer it suffices to restrict attention to an $(n+M)$-dimensional subspace spanned by radial and polynomial functions. Furthermore, we show that, for well-chosen loss functions, optimal rates of convergence may be attained without placing any moment conditions on the errors, thereby allowing for very heavy tailed error distributions within our theoretical analysis.

Thin-plate splines possess a number of notable advantages over their more popular tensor product counterparts. First, the objective function in \eqref{eq:3} only involves one smoothing parameter irrespective of the dimension of the predictor space. By contrast, tensor product penalties require as many smoothing parameters as the number of predictors. Since smoothing parameters are usually selected from the data, tensor product smoothers entail a considerable computational burden which becomes prohibitive for higher dimensions. Another attractive property of thin-plate splines is the invariance of the penalty $I_m^2(f)$ to rotations or reflections of the coordinate axes. This fact implies that thin-plate splines are ideal for spatial/geographic smoothing as in these cases the amount of smoothing does not depend on which axes represent the relative positions, e.g., latitude and longitude.  To these advantages we may add that tensor product smooths often rely on the subjective choice of the number of knots and their position, whereas with thin-plate splines the user is absolved from this responsibility, as both of these aspects emerge naturally from the mathematical problem that thin-plate splines solve.

In view of these advantages it comes as a surprise that thin-plate splines, with the exception of univariate splines with $d=1$, have not received enough attention in the literature. In fact, both available treatments of thin-plate spline estimators of arbitrary dimensions \citep{Cox:1984, G:2010} concern the narrow class of least squares thin-plate estimators and as a result the larger and more versatile class of M-type thin-plate estimators has been overlooked. To address this gap in the literature, the present paper establishes optimal rates of convergence for a wide variety of thin-plate estimators under a mild set of conditions considerably extending our theoretical understanding. Furthermore, while the main emphasis of this work is on outlier-resistant loss functions, we also demonstrate how our methodology can be used to improve upon the results of \citet{G:2010} by sharpening the rate of convergence obtained by these authors for the least squares thin-plate estimator. 

\section{Main results: existence of solutions and rates of convergence}

We begin our study of generalized M-type thin-plate estimators by providing sufficient conditions for the existence of a minimizer in $\mathcal{H}^{m}(\mathbbm{R}^d)$ for the objective function $L_n(f)$, as given in \eqref{eq:3}. For the special least squares case with $\rho(x) = x^2/2$ such conditions have already been presented in the literature \citep[see, e.g., p. 31,][]{Wahba:1990} and involve the uniqueness of the least squares estimator on the null-space of the penalty functional $I_m$. This null-space is $M$-dimensional with $M =  \binom{m+d-1}{d}$  and consists of all polynomials of total degree at most $m$. Let $\phi_1, \ldots, \phi_M$ denote any basis for this space of polynomials. Proposition~\ref{prop:1} below shows that the least squares requirement for existence generalizes nicely to arbitrary convex losses.
\begin{prop}
\label{prop:1}
Assume that $\rho$ is a convex loss function, $2m>d$ and the covariates $\mathbf{x}_i \in \mathbb{R}^d$ are such that the corresponding unpenalized M-estimator on $\phi_1, \ldots, \phi_M$ restricted to the $\mathbf{x}_i$ is unique. Then, $L_n(f)$ has a minimizer in $\mathcal{H}^{m}(\mathbbm{R}^d)$.
\end{prop}
The uniqueness requirement of Proposition~\ref{prop:1} is very mild. For example, for $d = 2$ and $m=2$ we may take $\phi_1(x_1, x_2) = 1$, $\phi_2(x_1, x_2) = x_1$ and $\phi_3(x_1, x_2) = x_2$, so that if the estimated plane is unique then the existence of a minimizer is guaranteed. The condition $2m>d$ is satisfied for all $m \geq 1$ when $d=1$, i.e., univariate data, but precludes small values of $m$ in higher dimensions. Unfortunately, this condition cannot be weakened as it is both necessary and sufficient for $\mathcal{H}^m(\mathbbm{R}^d)$ to be a reproducing kernel Hilbert space and consequently for point evaluation maps $\mathcal{H}^m(\mathbbm{R}^d) \to \mathbbm{R} : f \mapsto f(\mathbf{x}), \  \mathbf{x} \in \mathbbm{R}^d$, to be well defined. Proposition~\ref{prop:1} is of primary importance not only as a first step towards the theoretical analysis of M-type thin-plate estimators but also for their efficient computation; Section~\ref{sec:3} provides the details. 

Having established that the problem is well-defined, we now investigate the rate of convergence of the minimizer of $L_n(f)$ in $\mathcal{H}^{m}(\mathcal{R}^d)$, which we denote by $\widehat{f}_n$, to the true function $f_0$. The distance metric to be used is the $\mathcal{L}^2(Q_n)$-distance given by
\begin{align*}
\|f-g\|_{\mathcal{L}^2(Q_n)} = \left\{ \int_{\mathbbm{R}^d} |f(\mathbf{x})-g(\mathbf{x})|^2 d Q_n(\mathbf{x}) \right\}^{1/2}
\end{align*}
where $Q_n$ is the empirical measure placing mass $n^{-1}$ at each $\mathbf{x}_i \in \mathbbm{R}^d$. As we shall see, under very mild conditions on the covariates, rates of convergence in terms of the empirical norm $\| \cdot \|_{\mathcal{L}^2(Q_n)}$ also translate to rates of convergence in the classical $\mathcal{L}^2$-distance. It is worth noting that, while here and in the sequel we treat the covariates as non-random, all subsequent results also hold for random covariates provided that we condition on them. Our theoretical development relies on the following regularity conditions.

\begin{enumerate}
\item[A1.] The loss function $\rho: \mathbbm{R} \to \mathbb{R}_{+}$ is convex and satisfies a Lipschitz condition, i.e., there exists a $C>0$ such that
\begin{align*}
|\rho(x)-\rho(y)| \leq C|x-y|, \quad \forall (x,y) \in \mathbbm{R}^2.
\end{align*}
\item[A2.] The errors $\epsilon_i, i = 1, \ldots, n$ are independent random variables and there exists a constant $\kappa >0$ such that for all $|t| \leq \kappa$,
\begin{align*}
\inf_{n} \min_{1 \leq i \leq n} \mathbb{E}\{ \rho(\epsilon_i+t) - \rho(\epsilon_i) \} \geq \kappa t^2
\end{align*}
\item[A3.] The covariates $\mathbf{x}_i, \ i= 1, \ldots, n$ are contained in a bounded open set $\mathcal{O} \subset \mathbb{R}^d$ whose boundary, $\partial \mathcal{O}$, satisfies the uniform cone condition of \citet[p. 83]{Adams:2003}. That is, there exists a locally finite open cover $\{U_j\}$ of $\partial \mathcal{O}$ and a corresponding sequence of finite cones $\{C_j\}$, each congruent to some fixed cone $C$, such that
\begin{itemize}
\item[(i)] There exists an $M<\infty$ such that every $U_j$ has diameter less than $M$.
\item[(ii)] $\{ \mathbf{x} \in \mathcal{O}: \inf_{\mathbf{y} \in \partial\mathcal{O}} \| \mathbf{x}-\mathbf{y} \| < \delta \} \subset \bigcup_{j=1}^{\infty} U_j$ for some $\delta>0$.
\item[(iii)] $Q_j \equiv \bigcup_{\mathbf{x} \in \mathcal{O} \cap U_j}(\mathbf{x} + C_j) \subset \mathcal{O} $ for every $j$.
\item[(iv)] For some finite $R$, every collection of $R+1$ of the sets $Q_j$ has an empty intersection.
\end{itemize}
\item[A4.] Define the quantities
\begin{align*}
h_{\max, n} & = \sup_{\mathbf{x} \in \mathcal{O}} \min_{1 \leq i \leq n} \|\mathbf{x}-\mathbf{x}_i\| 
\\ 
h_{\min, n} & = \min_{i \neq j} \| \mathbf{x}_i - \mathbf{x}_j\|.
\end{align*}
For all large $n$, there exist finite positive constants $B_1, B_2$ such that $h_{max,n} \leq B_1$ and $h_{\max, n}/h_{\min, n} \leq B_2$.
\end{enumerate}

Assumption A1 is very general and a large number of loss functions fulfils these conditions. The convexity requirement implies that $\rho$ is minimally continuous, but, importantly, differentiability is not needed for Theorem~\ref{thm:1} below to hold. Thus, this assumption also covers non-smooth loss functions, such as the quantile loss $\rho_{\alpha}(x)= x(\alpha - \mathcal{I}(x<0)),\ \alpha \in (0,1)$. Assumption A2 requires the local quadratic behaviour of $m_i(t) : = \mathbb{E}\{ \rho(\epsilon_i + t) \}$, $i= 1, \ldots, n$ about zero. Assumptions of this type have been extensively used in the asymptotics of M-estimators  \citep[see, e.g.,][Theorem 3.2.5, p. 289]{VDV:1996}. It is similarly a weak condition that is satisfied quite generally. To see this, observe that by Fubini's theorem
\begin{align*}
m_i(t) - m_i(0) = \int_{0}^t \mathbb{E}\{\rho^{\prime}(\epsilon_i+x)\} dx,
\end{align*}
where $\rho^{\prime}$ is any subgradient of $\rho$, so that all examples of loss functions given by \citet{Kalogridis:2021} also satisfy our A2 under the conditions given by that author. Our assumptions do not entail identically distributed errors, as in our experience this is too strong of an assumption for many practical settings. 

Conditions A3 and A4 concerning the design points (knots) and their positions have been previously used in the asymptotics of least squares thin-plate spline estimators, see \citet{Utr:1988}. The uniform cone condition precludes sets with very irregular boundaries, but is satisfied quite generally otherwise; it is satisfied, e.g., by balls and rectangles. Condition A4 is also very modest, as it essentially requires the design points to be distinct and dense within the domain of interest, at least for large $n$. Both of these requirements follow from the fact that the ratio $h_{\max, n}/h_{\min, n}$ remains bounded for all large $n$, as if either the observations are not unique or dense enough as $n \to \infty$, $h_{\max, n}/h_{\min, n}$ could become unbounded.

With these assumptions we can now state our first asymptotic result. It is worth noting that in Theorem~\ref{thm:1} below we treat the smoothing parameter $\lambda$ as a random variable possibly depending on our sample $(\mathbf{x}_1, Y_1), \ldots, (\mathbf{x}_n, Y_n)$. We place no restrictions on the way that $\lambda$ may depend on the sample and instead merely require that $\lambda$ is a measurable function of our sample so that $\lambda$ is a properly defined random variable. Our treatment of $\lambda$ needs to be contrasted with the treatment of the smoothing parameter by other authors, e.g., \citet{Cox:1984, Utr:1988}, who regard it as a deterministic sequence. In our view, our treatment constitutes an important extension of existing results as $\lambda$ is most often selected from the data and is thus random rather than fixed, see, e.g., the data-driven method of selecting $\lambda$ presented in Section~\ref{sec:3}.

\begin{theorem}
\label{thm:1}
Assume that the conditions of Proposition~\ref{prop:1} are met, A1--A4 hold and further that $\lambda = O_P(n^{-2m/(2m+d)})$ as well as $\lambda^{-1} = O_P(n^{2m/(2m+d)})$. Then, there exists a sequence, $\widehat{f}_n$, of M-type thin-plate splines minimizing \eqref{eq:3} such that
\begin{align*}
\|\widehat{f}_n-f_0\|_{\mathcal{L}^2(Q_n)}^2 = O_P(n^{-2m/(2m+d)}) \quad \text{and} \quad I_m(\widehat{f}_n) = O_P(1),
\end{align*}
as $n \to \infty$.
\end{theorem}
\noindent
The limit conditions involving the smoothing parameter require that $\lambda$ tends to zero in probability as $n \to \infty$, but not too fast. The rate of decay $n^{-2m/(2m+d)}$ for $
\lambda$ ensures that the asymptotic variance and bias of the estimator are balanced and this leads to the rate of convergence $n^{-2m/(2m+d)}$ for thin-plate estimators. It is worth noting that $n^{-2m/(2m+d)}$ is the optimal (squared) rate of convergence for functions in $\mathcal{H}^{m}(\mathcal{O})$ \citep{Stone:1982} and therefore it cannot be improved, except in trivial cases. For $d=1$ our rate of convergence is in agreement with the rate obtained by \citet{Kalogridis:2021} for univariate smoothing spline estimators, but for $d>1$ the result in Theorem~\ref{thm:1} constitutes an important generalization. The obtained rate of convergence suggests that thin-plate spline estimators, like most nonparametric estimators, suffer from the curse of dimensionality and consequently, for given sample size $n$, estimation becomes less and less precise for larger predictor dimension $d$.

Theorem~\ref{thm:1} establishes not only a rate of convergence in the empirical norm $\| \cdot \|_{\mathcal{L}^2(Q_n)}$, but also the boundedness of the semi-norm $I_m(\widehat{f}_n)$. The latter is crucial in extending this rate of convergence to the $\mathcal{L}^2$-norm as well as in establishing optimal rates of convergence of certain useful derivatives. These extensions are given in Corollary~\ref{cor:1} below.
\begin{cor}
\label{cor:1}
Suppose that the assumptions of Theorem~\ref{thm:1} hold and that $h_{\max, n} = O(n^{-1/(2m+d)})$. Then, the sequence of minimizers of \eqref{eq:3}, $\widehat{f}_n$, satisfies
\begin{align*}
\int_{\mathcal{O}} |\widehat{f}_n(\mathbf{x})-f_0(\mathbf{x})|^2 d\mathbf{x} = O_P(n^{-2m/(2m+d)}),
\end{align*}
and for every $ (j_1, \ldots, j_d) \in \mathbbm{R}^d$ such that $j_1+\ldots+j_d = j \leq m$ 
\begin{align*}
\int_{\mathcal{O}} \left| \frac{\partial \widehat{f}_n^j(\mathbf{x})}{\partial x_1^{j_1} \ldots \partial x_d^{j_d} } -  \frac{\partial f^j_0(\mathbf{x})}{\partial x_1^{j_1} \ldots \partial x_d^{j_d} } \right|^2 d \mathbf{x} = O_P(n^{-2(m-j)/(2m+d)}).
\end{align*} 
\end{cor}
\noindent
The above rates of convergence are again optimal according to the results of \citet{Stone:1982}. To the best of our knowledge, these are the first convergence results for derivatives of thin-plate estimators even in the relatively simple least squares case. 

It is interesting to compare the rate of convergence $n^{-2m/(2m+d)}$ obtained herein with the rate $ \log(n) n^{-2m/(2m+d)}$ for the least squares thin-plate spline estimator emerging from the results of \citet[Chapter 21]{G:2010}. The square loss $\rho(x) = x^2/2$ is not covered by our previous set of assumptions (it is not Lipschitz), but least squares estimators may be easily treated on a separate basis. In particular, letting $\widehat{f}_n$ now denote the least squares estimator, one may easily verify the inequality
\begin{align}
\label{eq:4}
\|\widehat{f}_n - f_0 \|^2_{\mathcal{L}^2(Q_n)} + \lambda I_m^2(\widehat{f}_n) \leq 2 \langle \epsilon, \widehat{f}_n - f_0 \rangle_{\mathcal{L}^2(Q_n)} + \lambda I_m^2(f_0),
\end{align}
where $\langle \epsilon, \widehat{f}_n - f_0  \rangle_{\mathcal{L}^2(Q_n)}$ stands for $n^{-1} \sum_{i=1}^n \epsilon_i (\widehat{f}_n(\mathbf{x}_i) - f_0(\mathbf{x}_i))$. This inequality combined with our sharper estimate for the local entropy (compare our Lemma~\ref{lem:1} below with Lemma 20.6 in \citet{G:2010}) and the modulus of continuity of the empirical process derived in \citet[Chapter 10]{van de Geer:2000} now leads to Theorem~\ref{thm:2}.

\begin{theorem}
\label{thm:2}
Assume that the conditions of Proposition~\ref{prop:1} are met, A3 holds and that the errors $\epsilon_i$ are uniformly sub-Gaussian, i.e., there exist finite constants $K_1$ and $K_2 >0$ such that
\begin{align*}
\sup_n \max_{1 \leq i \leq n} K_1^2 \mathbb{E}\{e^{|\epsilon_i|^2/K_1^2}-1 \} \leq K_2.
\end{align*}
If $\lambda = O_P(n^{-2m/(2m+d)})$ and $\lambda^{-1} = O_P(n^{2m/(2m+d)})$, then there exists a sequence of least squares thin-plate spline estimators, $\widehat{f}_n$, satisfying
\begin{align*}
\|\widehat{f}_n-f_0\|_{\mathcal{L}^2(Q_n)}^2 = O_P(n^{-2m/(2m+d)}) \quad \text{and} \quad I_m(\widehat{f}_n) = O_P(1),
\end{align*}
as $n \to \infty$.
\end{theorem}
\noindent
It is important to emphasize that the sub-Gaussian requirement in Theorem~\ref{thm:2} is used exclusively for the treatment of the least-squares thin-plate estimator and not for the robust estimators for the treatment of which we rely solely on assumptions A1--A4. The sub-Gaussian requirement is met in practice, e.g., whenever the errors follow a Gaussian distribution or, more generally, whenever they possess a squared exponential moment. As noted previously, Theorem~\ref{thm:2} improves upon the corresponding result of \citet{G:2010}. Rates of convergence for the least squares estimator in the $\mathcal{L}^2(\mathcal{O})$-norm as well as rates of convergence for the derivatives may now be established exactly as in the proof of Corollary~\ref{cor:1}; we omit the details.

\section{Computation and smoothing parameter selection}
\label{sec:3}

As hinted previously, Proposition~\ref{prop:1} is crucial not only from a theoretical, but also from a practical standpoint. In fact, with the help of Proposition~\ref{prop:1} and reasoning along the same lines as in the proof of \citet[Theorem 7.3]{Green:1994}, it can be shown that in order to identify the minimizer of $L_n(f)$ in \eqref{eq:3} it suffices to restrict attention to functions of the form
\begin{align}
\label{eq:5}
f(\mathbf{x}) = \sum_{i=1}^n \gamma_i \eta_{m, d}(\|\mathbf{x}-\mathbf{x}_i\|) + \sum_{j=1}^M \delta_j \phi_j(\mathbf{x}),
\end{align}
where $\eta_{m,d}:\mathbb{R}_{+} \to \mathbb{R}$ is given by
\begin{align*}
\eta_{m, d}(x) = \begin{cases} \frac{(-1)^{m+1+d/2}}{2^{2m-1} \pi^{d/2} (m-1)!(m-d/2)!} x^{2m-d} \log(x) & d \ \text{even} \\ 
\frac{\Gamma(d/2-m)}{2^{2m} \pi^{d/2} (m-1)!}x^{2m-d} & d \ \text{odd}, \end{cases}
\end{align*}
with $\Gamma(\cdot)$ denoting Euler's gamma function. The coefficient vector $\boldsymbol{\gamma}$ is subject to the set of linear constraints $\boldsymbol{\Phi}^{\top}\boldsymbol{\gamma} = \mathbf{0}$, where the $n \times M$ matrix $\boldsymbol{\Phi}$ has $(i,j)$th entry $\phi_j(\mathbf{x}_i)$. In other words, the coefficient vector $\boldsymbol{\gamma} \in \mathbbm{R}^n$ needs to be perpendicular to the $M$-dimensional space spanned by the restriction of the polynomial functions $\phi_1, \ldots, \phi_M$ to the design points. 

The representation in \eqref{eq:5} as well as the set of linear constraints are derived from reproducing kernel Hilbert space arguments that split the space $\mathcal{H}^m(\mathbb{R}^d)$ into the closed finite-dimensional null space of $I_m(f)$ and its orthogonal complement. It follows from these arguments that  plugging \eqref{eq:5} into \eqref{eq:3} yields the quadratic form $I_m^2(f) = \boldsymbol{\gamma}^\top \mathbf{\Omega}\boldsymbol{\gamma} $ where the $n \times n$ matrix $\mathbf{\Omega}$ has $(i,j)$th entry $\eta_{m,d}(\|\mathbf{x}_i - \mathbf{x}_j\|)$ and the minimization problem becomes
\begin{equation}
\begin{aligned}
\label{eq:6}
& \min_{(\boldsymbol{\gamma}, \boldsymbol{\delta}) \in \mathbbm{R}^n \times \mathbbm{R}^M}  \left[ \frac{1}{n} \sum_{i=1}^n \rho(Y_i - \boldsymbol{\omega}_i^{\top} \boldsymbol{\gamma}  - \boldsymbol{\phi}_i^{\top} \boldsymbol{\delta}) + \lambda \boldsymbol{\gamma}^\top \mathbf{\Omega}\boldsymbol{\gamma}\right] \\ 
& \quad \quad \text{s.t.} \  \boldsymbol{\Phi}^{\top} \boldsymbol{\gamma} = \mathbf{0},
\end{aligned}
\end{equation}
with $\boldsymbol{\omega}_i \in \mathbbm{R}^n$ and $\boldsymbol{\phi}_i \in \mathbbm{R}^M$ denoting row vectors of $\mathbf{\Omega}$ and $\boldsymbol{\Phi}$, respectively. It can be shown that for all $\boldsymbol{\gamma} \in \mathbbm{R}^n$ satisfying $\boldsymbol{\Phi}^{\top}\boldsymbol{\gamma} = \mathbf{0}$ we have $\boldsymbol{\gamma}^{\top} \mathbf{\Omega} \boldsymbol{\gamma}>0$, see \citet[pp. 32--33]{Wahba:1990} and the references therein. We may automatically incorporate the linear constraints and further simplify the problem by putting $\boldsymbol{\gamma} = \mathbf{Q}\boldsymbol{\beta}$ for a matrix $\mathbf{Q}$ whose columns span the null space of $\boldsymbol{\Phi}^{\top}$ and $\boldsymbol{\beta} \in \mathbbm{R}^{n-M}$. The matrix $\mathbf{Q}$ can be easily obtained through the QR or singular value decompositions of $\boldsymbol{\Phi}$. With this simplification, the solution to \eqref{eq:6} may be identified using, for example, the variant of the iteratively reweighted least squares (IRLS) algorithm employed by \citet{Kalogridis:2021} with the radial basis functions and polynomials in \eqref{eq:5} taking the place of B-spline basis functions there. 

To determine the penalty parameter $\lambda$ in a data-driven way, we propose a suitable adaptation of the strategy of \citet{Maronna:2011}. Let $\mathbf{r}_{-}=(r_{-1}, \ldots, r_{-n})^{\top}$ denote an approximation to the leave-one out residuals, as obtained, for example, from the last step of the IRLS algorithm. We propose to select the value of $\lambda$ that minimizes the robust cross-validation (RCV) criterion
\begin{equation*}
\RCV(\lambda) =  |\tau(\mathbf{r}_{-})|^2,
\end{equation*}
where $\tau$ denotes the robust and efficient $\tau$-scale introduced by \citet{Yohai:1988} with tuning constants equal to $c_1 = 3$ and $c_2 = 5$, corresponding to the biweighting of the mean and standard deviation respectively. This criterion may be viewed as a robustification of the celebrated leave-one-out criterion \citep[see, e.g.,][pp. 47--52]{Wahba:1990} in which the $\tau$-scale is replaced by the mean of squares of the $r_{-i}$. Thus, while in the classical leave-one-out criterion all the $|r_{-i}|^2$ contribute equally (with weight $n^{-1}$ each), the use of the robust $\tau$-scale employed herein reduces the effect of large $|r_{-i}|^2$ on the selection of $\lambda$ thereby leading to a robust automatic selection procedure.

\section{A Monte Carlo study}
\label{sec:4}

We now examine the practical performance of several thin-plate spline estimators by means of a simulation study. We are interested in the performance of the competing estimators not only in completely regular data, but also in data that may contain a number of outlying observations. The estimators to be considered are 
\begin{itemize}
\item The classical least squares thin-plate spline estimator, abbreviated as LS.
\item The least absolute deviations type thin-plate spline estimator with loss function $\rho(x) = |x|$ abbreviated as LAD.
\item The Huber type thin-plate spline estimator with loss function
\begin{align*}
\rho(x) = \begin{cases} x^2/2 & |x| < 1.345 \\ 1.345|x| - 1.345^2/2 & |x| \geq 1.345 \end{cases},
\end{align*}
abbreviated as Huber.
\item The logistic type thin-plate spline estimator with loss function 
\begin{align*}
\rho(x) = 2x + 4\log(1+e^{-x}),
\end{align*}
abbreviated as Logistic.
\end{itemize}

To compare the above estimators we generate observations from the model
\begin{align*}
Y_i = f_0(\mathbf{x}_i) + \epsilon_i, \quad(i=1, \ldots, n),
\end{align*}
where the regression function $f_0(\mathbf{x})$ is either given by $f_1(x_1, x_2) = \exp[-8|x_1-0.5|^2-8|x_2-0.5|^2]$, $f_2(x_1, x_2) = \sin(2 \pi x_1) \cos(2 \pi x_2)$ or $f_3(x_1, x_2, x_3) = x_1^2 + x_2^2 + x_3^2$, $(x_{i1}, x_{i2}, x_{i3})$ are uniformly distributed random variables in the unit square $(0,1)^3$ and the errors $\epsilon_i$ are generated according to the following distributions: (i) the standard Gaussian distribution, (ii) the t-distribution with $3$ degrees of freedom, (iii) the skewed t-distribution with $3$ degrees of freedom and non-centrality parameter equal to $1$ leading to right-skewed data, (iv) a mixture Gaussian distribution with weights means equal to $1$ and $10$, variances equal to $1$ and $0.01$ and weights equal to $0.85$ and $0.15$ respectively and (v) Tukey's Slash distribution defined as the distribution of the quotient of independent standard normal and uniform random variables. We assess the performance of the competing estimators through the mean-squared error (MSE) given by
\begin{align*}
\MSE = \frac{1}{n} \sum_{i=1}^n |\widehat{f}_n(\mathbf{x}_i)-f_0(\mathbf{x}_i)|^2,
\end{align*}
which is an approximation to the squared $\mathcal{L}^2$-error. Table~\ref{tab:1} below reports average MSEs and their standard errors obtained from 1000 replications with datasets of size $n=100$.

\begin{table}[H]
\centering
\begin{tabular}{lcccccccccc}
& & \multicolumn{2}{c}{LS}  &  \multicolumn{2}{c}{Huber} &  \multicolumn{2}{c}{Logistic} &  \multicolumn{2}{c}{LAD} \\ \\[-2ex]
$f_0$ & Dist. & Mean & SE & Mean & SE & Mean & SE & Mean & SE \\ \\[-1.2ex]
\multirow{5}{*}{$f_1$} & Gaussian & 7.17 & 0.11 & 7.34 & 0.10 & 7.39 & 0.10 & 11.76 & 0.14 \\ \\[-2ex] 
& $t_3$ & 15.11 & 0.36 & 9.52 & 0.13 & 9.89 & 0.15 & 13.01 & 0.19 \\ \\[-2ex]
& $st_{3, 1}$ &  74.04 & 0.87 & 52.26 & 0.49 & 56.92 & 0.55 & 48.76 & 0.52 \\ \\[-2ex]
& M. Gaussian & 292.2 & 4.22 & 21.49 & 0.40 & 36.41 & 0.68 & 21.06 & 0.39 \\ \\[-2ex]
& Slash & 9e+07 & 3e+06 & 25.09 & 0.48 & 27.16 & 0.56 & 27.80 & 0.56 \\ \\[-1.2ex]
\multirow{5}{*}{$f_2$} & Gaussian & 14.20 & 0.16 & 14.17 & 0.17 & 13.88 & 0.17 & 21.04 & 0.23 \\ \\[-2ex]
& $t_3$ & 28.25 & 0.45 & 21.98 & 0.37 & 22.02 & 0.30 & 24.44 & 0.27 \\ \\[-2ex]
& $st_{3, 1}$ & 84.83 & 0.90 & 62.13 & 0.55 & 66.81 & 0.58 & 59.13 & 0.57 \\ \\[-2ex]
& M. Gaussian & 311.1 & 4.18 & 39.12 & 1.23 & 57.14 & 1.50 & 34.94 & 0.89 \\ \\[-2ex]
& Slash & 4e+07 &  2e+06 & 72.19 & 2.63 & 56.70 & 1.38 & 57.29 & 0.95 \\ \\[-1.2ex]
\multirow{5}{*}{$f_3$} & Gaussian & 7.72 & 0.27 & 7.77 & 0.26 & 7.73 & 0.25 &  15.07 &  0.28 \\ \\[-2ex]
& $t_3$ & 18.19 & 0.75 & 11.52 & 0.43 & 12.01 & 0.43 & 18.43 & 0.33 \\ \\[-2ex]
& $st_{3, 1}$ & 74.37 & 1.05 & 53.10 & 0.57 & 57.45 & 0.60 & 52.47 & 0.56 \\ \\[-2ex]
& M. Gaussian & 321.90 & 6.35 & 25.89 & 0.72 & 43.02 & 1.47 & 25.09 & 0.50 \\ \\[-2ex]
& Slash & 4e+07 &  5e+06 & 45.83 & 1.40 & 48.63 & 1.73 & 43.66 & 0.87
\end{tabular}
\caption{Means and standard errors of 1000 MSEs $(
\times 100)$ with $n=100$ of the least squares, Huber, logistic and least absolute deviations type thin-plate spline estimators.}
\label{tab:1}
\end{table}

The results in Table~\ref{tab:1} indicate the sensitivity of least squares estimator to departures from the (sub-)Gaussian assumption. In particular, even with $t_3$ errors that lead to a few mild outliers, the least squares estimator significantly loses ground compared to the Huber, Logistic and LAD estimators. The former two estimators almost match the performance of the least squares estimator in regular data but also exhibit a high degree of resistance towards gross errors. The LAD estimator is inefficient relative to its competitors in the situation of light-tailed Gaussian errors, but exhibits superior performance to the Huber and Logistic thin-plate spline estimators in situations of  heavy-tailed asymmetric contamination, such as contamination incurred by $st_{3,1}$ and mixture Gaussian errors. 

To illustrate the key practical differences between the least squares thin-plate spline and its robust counterparts, Figures \ref{fig1} and \ref{fig2} present two examples of estimated surfaces by the least squares and least absolute deviations estimators under Gaussian and mixture Gaussian errors, respectively. The regression function $f_1$ has a unit-sized bump at $(0.5, 0.5)$, which both estimators get right in the absence of outliers, as evidenced in Figure~\ref{fig1}. Figure~\ref{fig2}, however, indicates that least squares estimates can be heavily distorted by the presence of outlying observations to the extent that the estimated surface bears no resemblance to the true surface. Despite the heavy contamination, the bump is still visible in the right panel of Figure~\ref{fig2}, which depicts the LAD estimated surface. 

\begin{figure}[H]
\centering
\subfloat{\includegraphics[ width = 0.49\textwidth]{"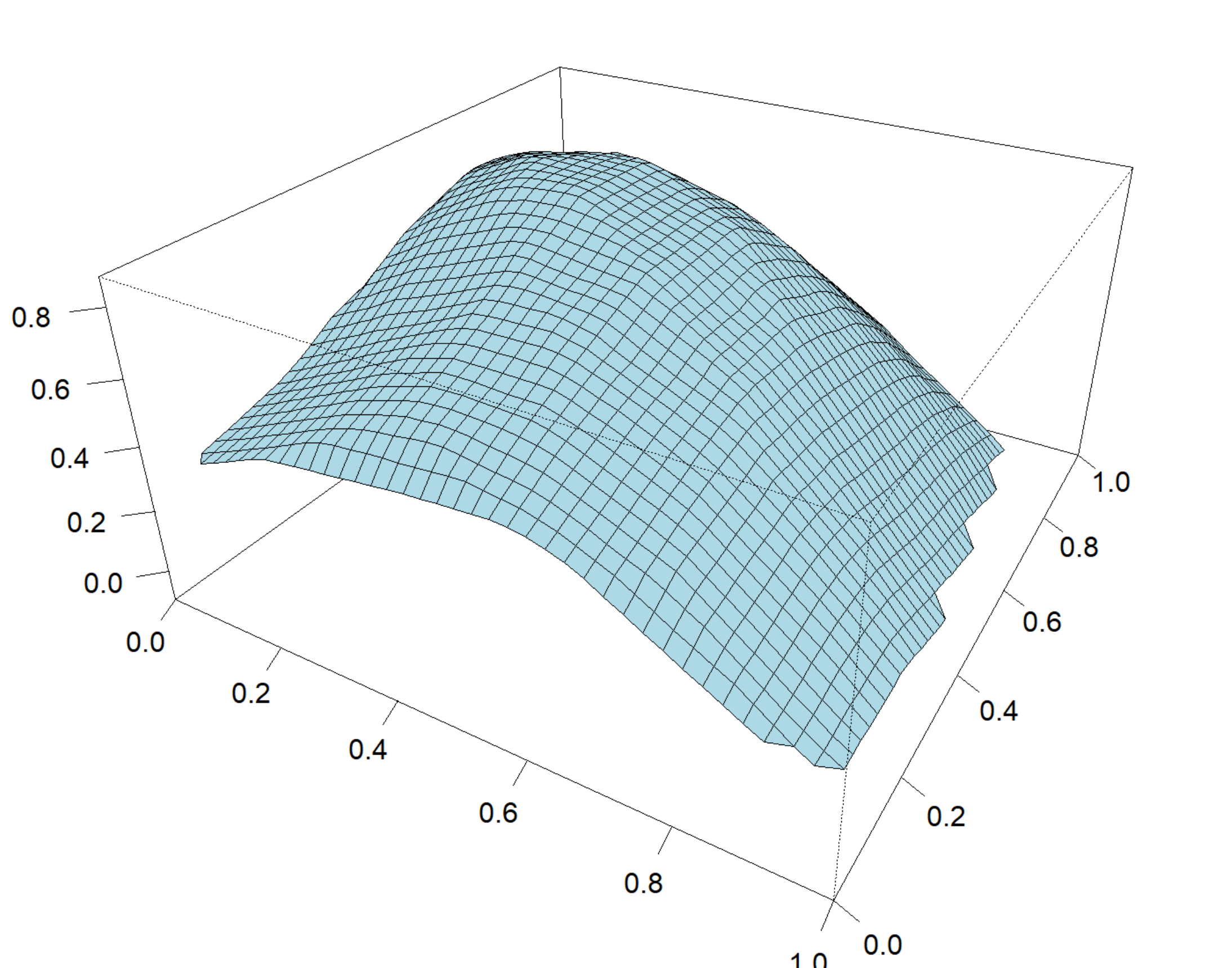"}}  
\subfloat{\includegraphics[ width = 0.49\textwidth]{"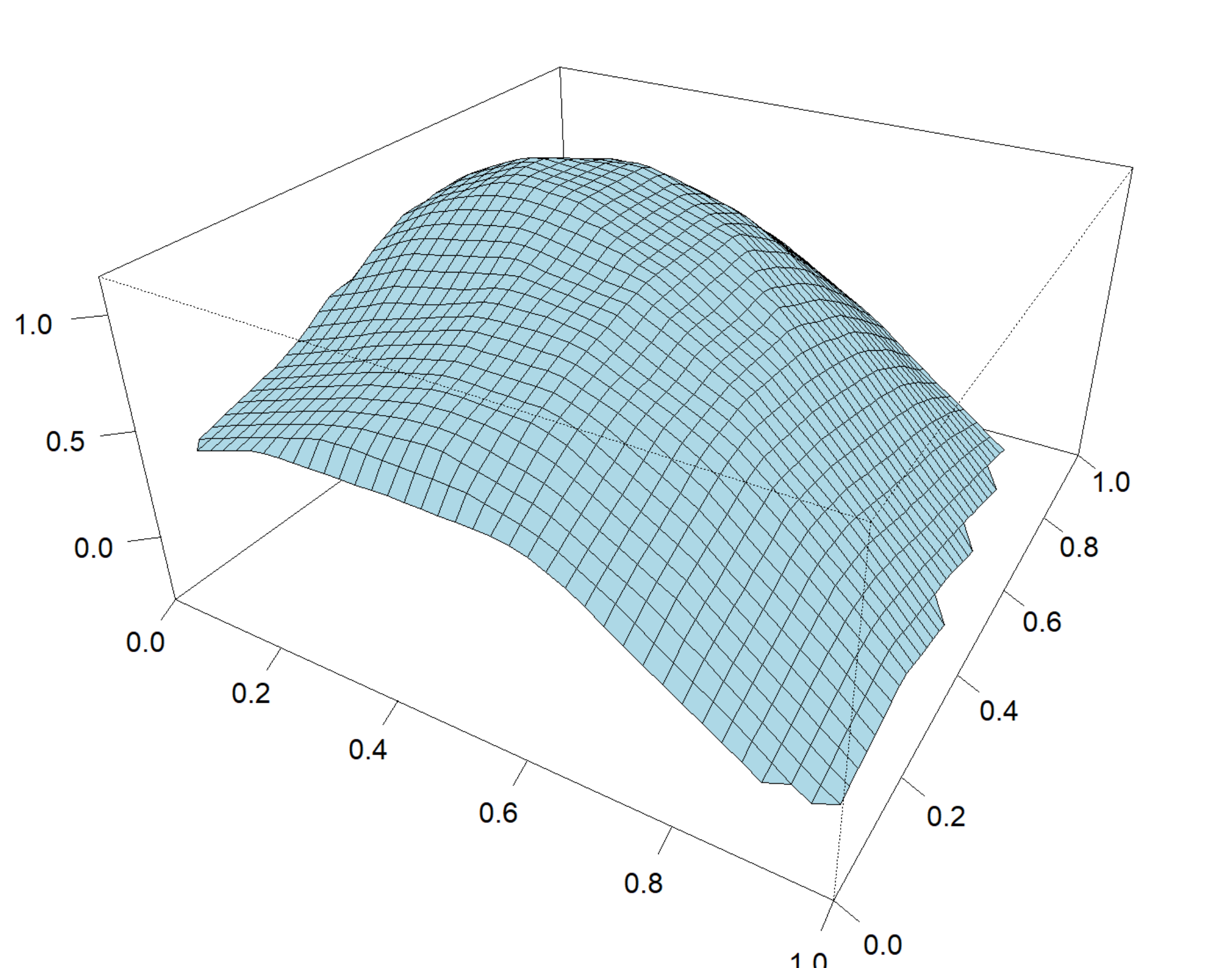"}} 
\caption{Typical estimated surfaces for $f_1(x_1, x_2) =  \exp[-8|x_1-0.5|^2-8|x_2-0.5|^2]$ under Gaussian errors by the least squares and least absolute deviations estimators on the left and right panels, respectively.}
\label{fig1}
\end{figure}

\section{Application: Ozone Levels in Midwestern USA}

While stratospheric ozone protects living organisms from ultraviolet radiation from the sun, high ground-level concentrations of ozone can trigger a variety of health problems, particularly for people with breathing difficulties, children and the elderly. In this example we examine the concentration of ground-level in the midwestern US as a function of geographical longitude and latitude. That is, we consider the model
\begin{align*}
\text{Ozone}_i = f_0(\text{Longitude}_i, \text{Latitude}_i) + \epsilon_i, \quad (i=1, \ldots, n).
\end{align*}
The data for this analysis consists of $8$-hour average surface ozone from 9am to 4pm in parts per billion (PPB) from 147 sites in midwestern US on July 3, 1987. The geographical coordinates of these sites constitute the predictor variables. The present dataset is part of a much larger dataset which is freely available as part of the \texttt{fields} package \citep{N:2017} in CRAN \citep{R}. 

\begin{figure}[H]
\centering
\subfloat{\includegraphics[ width = 0.49\textwidth]{"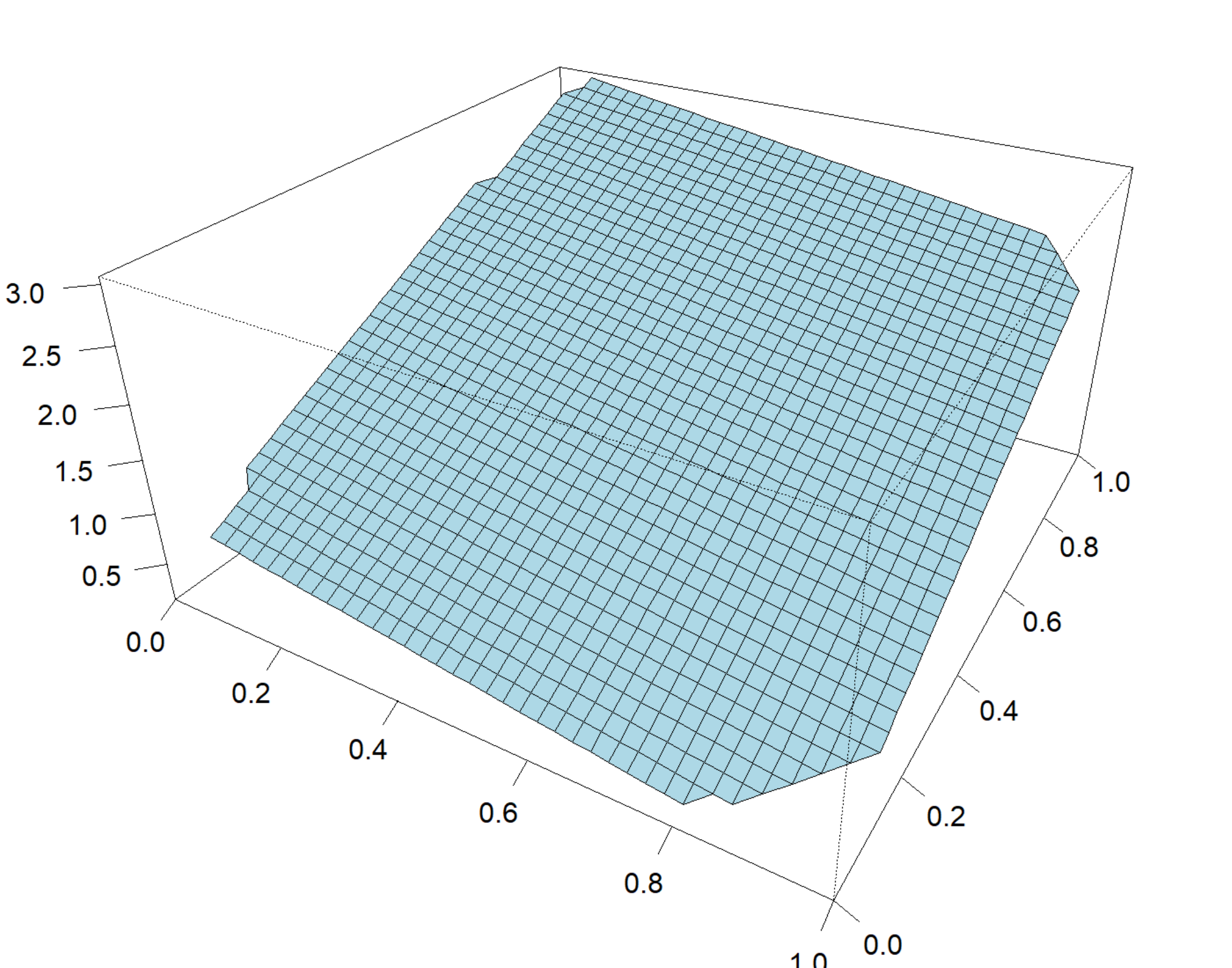"}}  
\subfloat{\includegraphics[ width = 0.49\textwidth]{"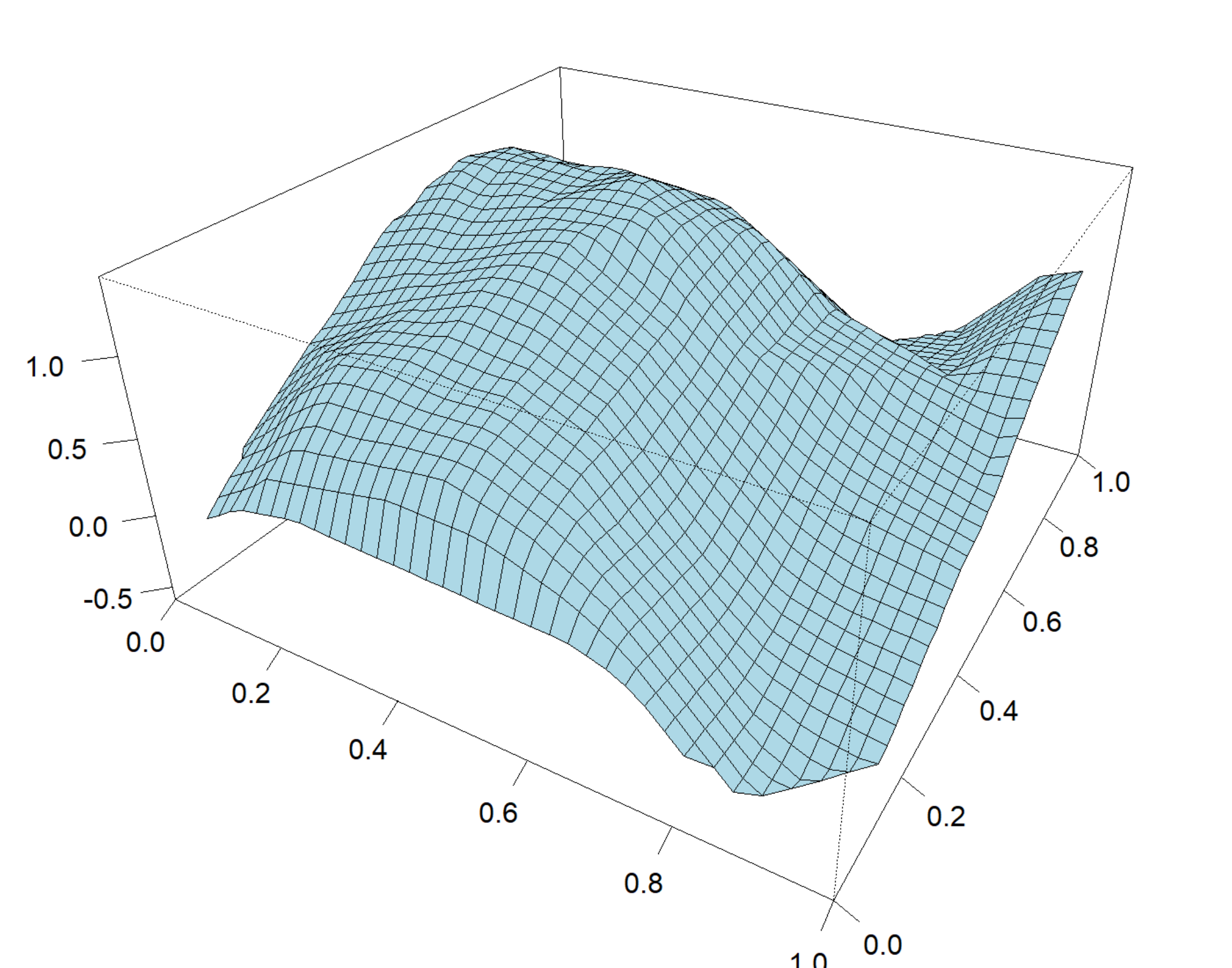"}} 
\caption{Typical estimated surfaces for $f_1(x_1, x_2) =  \exp[-8|x_1-0.5|^2-8|x_2-0.5|^2]$ under mixture Gaussian errors by the least squares and least absolute deviations estimators on the left and right panels, respectively.}
\label{fig2}
\end{figure}

Since ozone concentration tends to be a skewed random variable with several potentially outlying values, we have estimated the regression function $f_0$ with both the least squares (LS) and least absolute deviations (LAD) thin-plate spline estimators. The contours of the estimated surfaces on the convex hull of the data are depicted in the left and right panels of Figure~\ref{fig3}, respectively. The panels of the figure suggest that while there is broad agreement between these estimated surfaces around Indianapolis, the surfaces are noticeably different on the central and western part of the data. In particular, in the areas west and south of Milwaukee and St. Louis, LS estimates tend to underestimate ozone concentrations relative to LAD estimates. These differences are probably more consequential when it comes to concentrations in the excess of 60 PPB as these are more detrimental to one's health.

In view of the lack of resistance of LS estimators, it may be conjectured that the large differences between the estimates depicted in Figure~\ref{fig3} are attributable to the presence of atypical observations within the data. Since robust regression estimators are less attracted to outlying observations, such observations result in large absolute residuals which we may then use for their detection. A popular rule of thumb in that respect involves flagging the $i$th observation as an outlier if its standardized absolute LAD residual $r_i/\MAD(\mathbf{r})$ is larger than $2.5$. This rule results in the detection of 14 outlying observations for the LAD estimator, but only 7 for the non-robust LS estimator.

\begin{figure}[H]
\centering
\subfloat{\includegraphics[ width = 0.49\textwidth]{"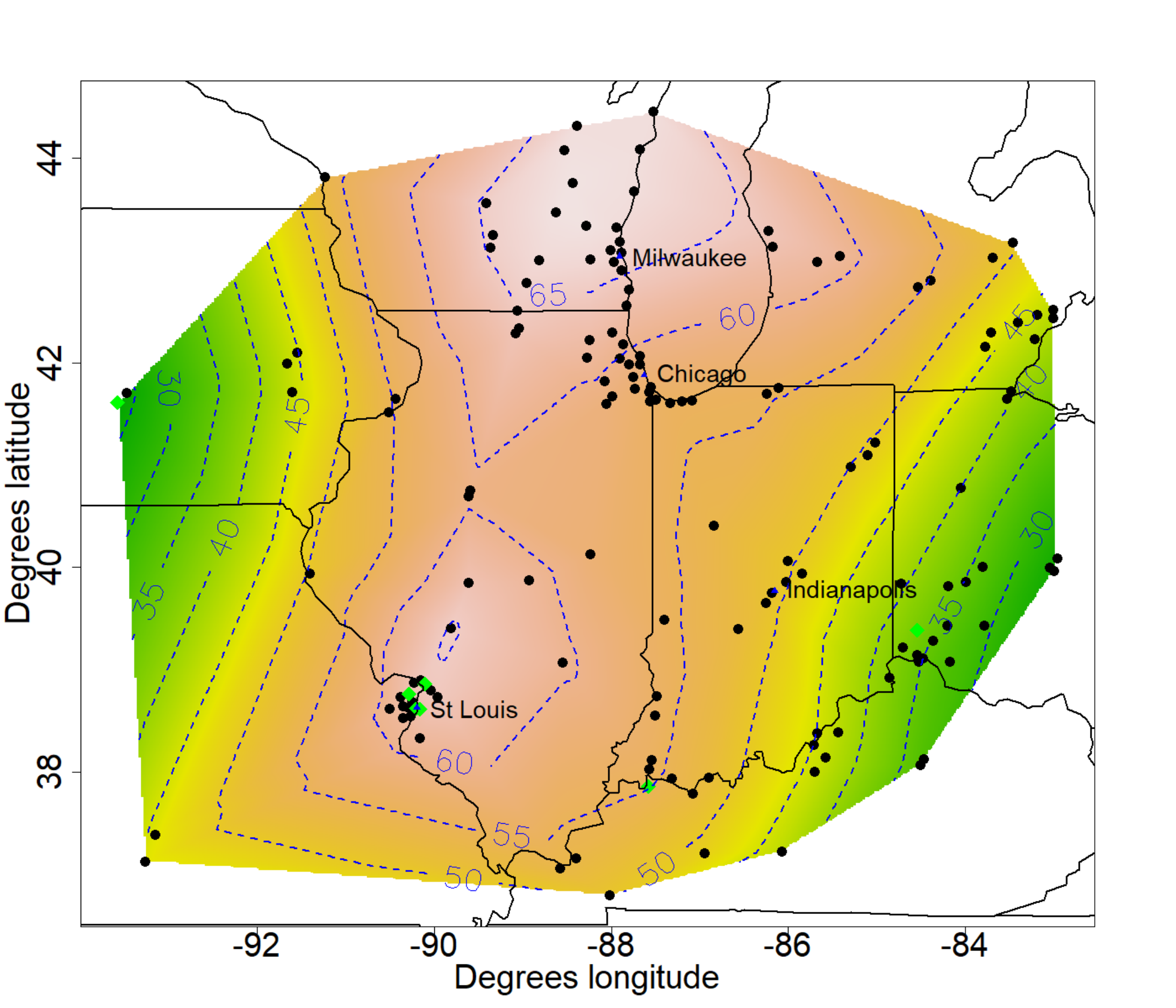"}}  
\subfloat{\includegraphics[ width = 0.49\textwidth]{"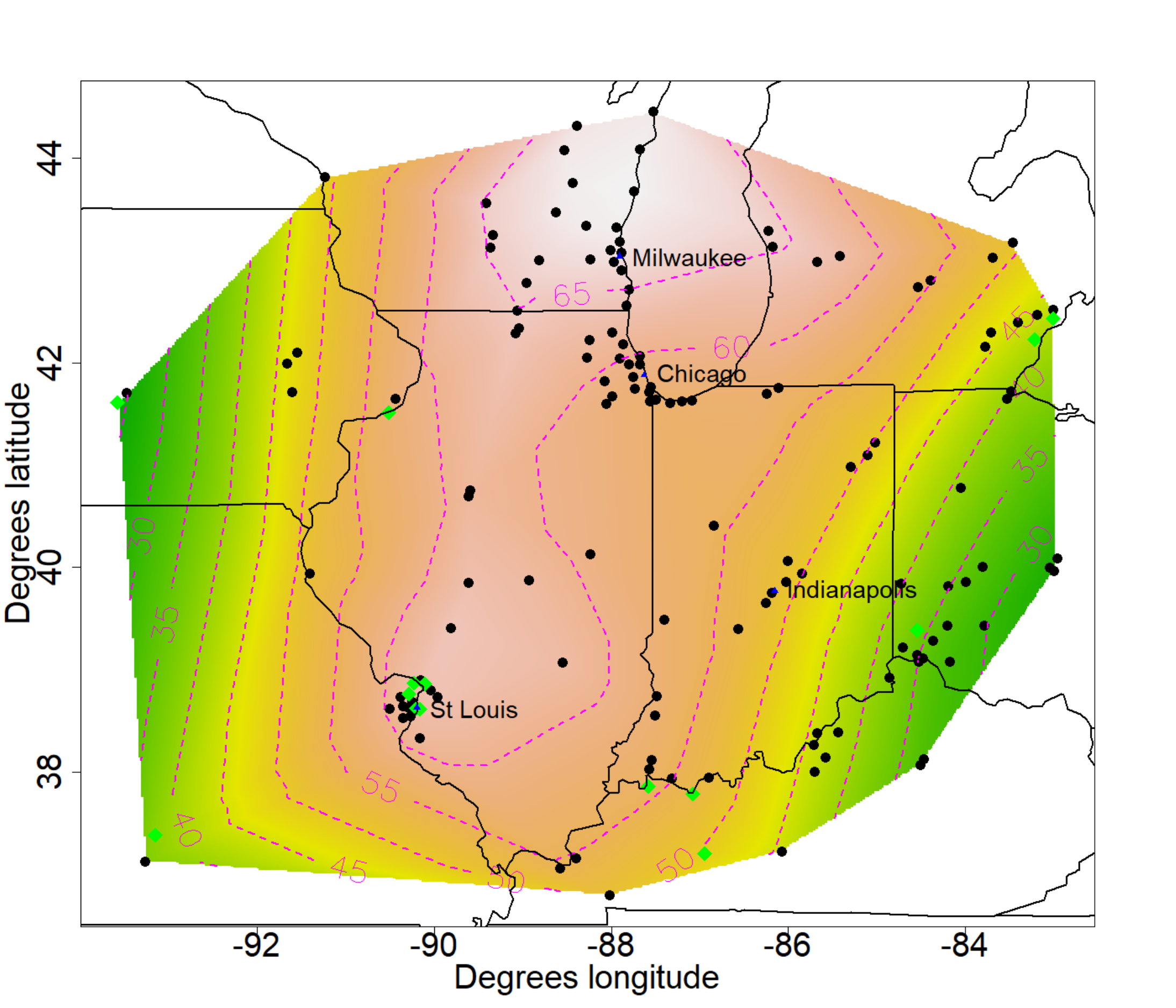"}} 
\caption{Contours of the estimated surfaces. Left: contours of the least squares thin-plate spline estimator. Right: contours of the least absolute deviations thin-plate spline estimator. Darker colors indicate higher ozone concentrations. The observation sites are depicted with solid black dots and green rhombuses depending on whether the observation is classified as an outlier or not.}
\label{fig3}
\end{figure}
\vspace{-0.2cm}
\begin{figure}[H]
\centering
\subfloat{\includegraphics[ width = 0.49\textwidth]{"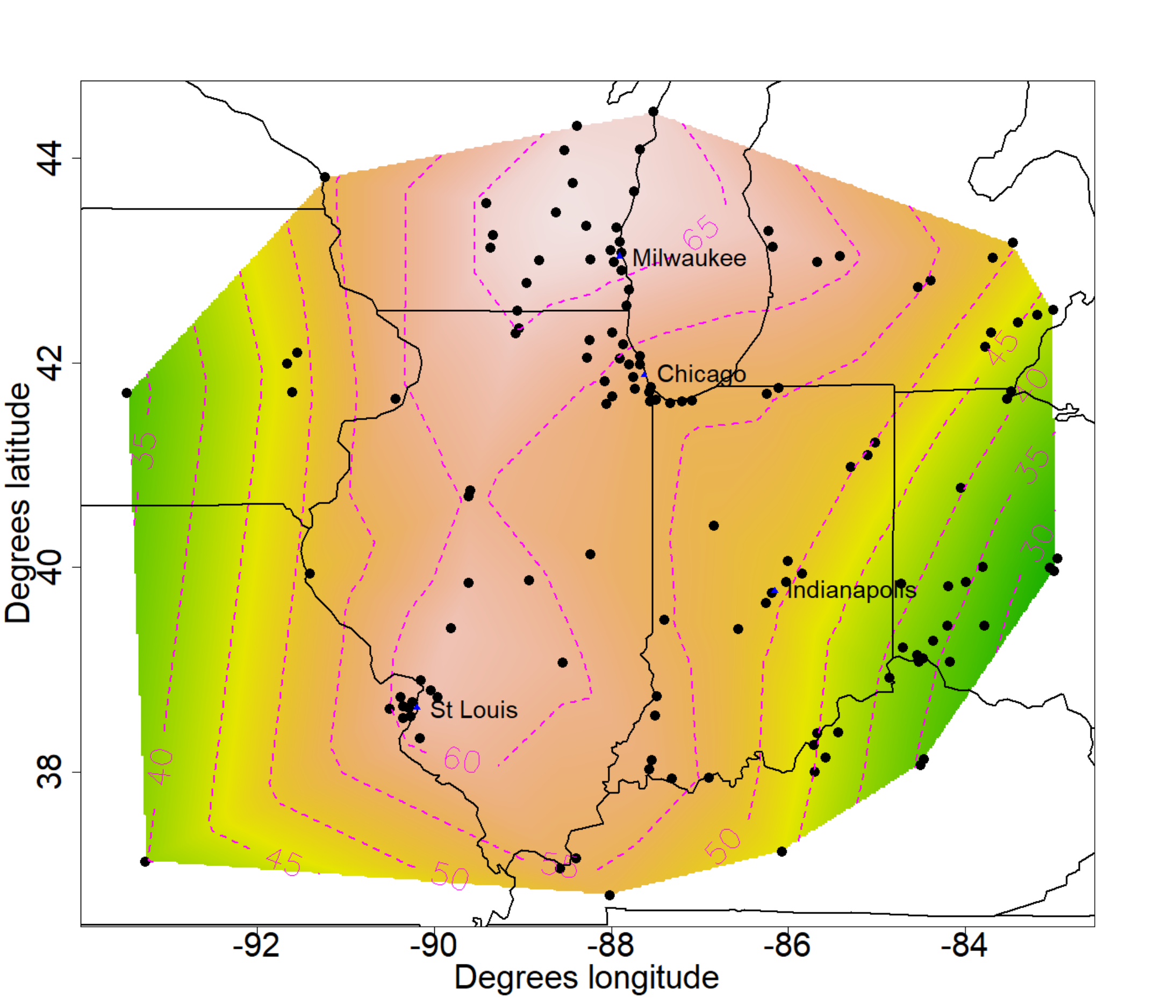"}}  
\subfloat{\includegraphics[ width = 0.49\textwidth]{"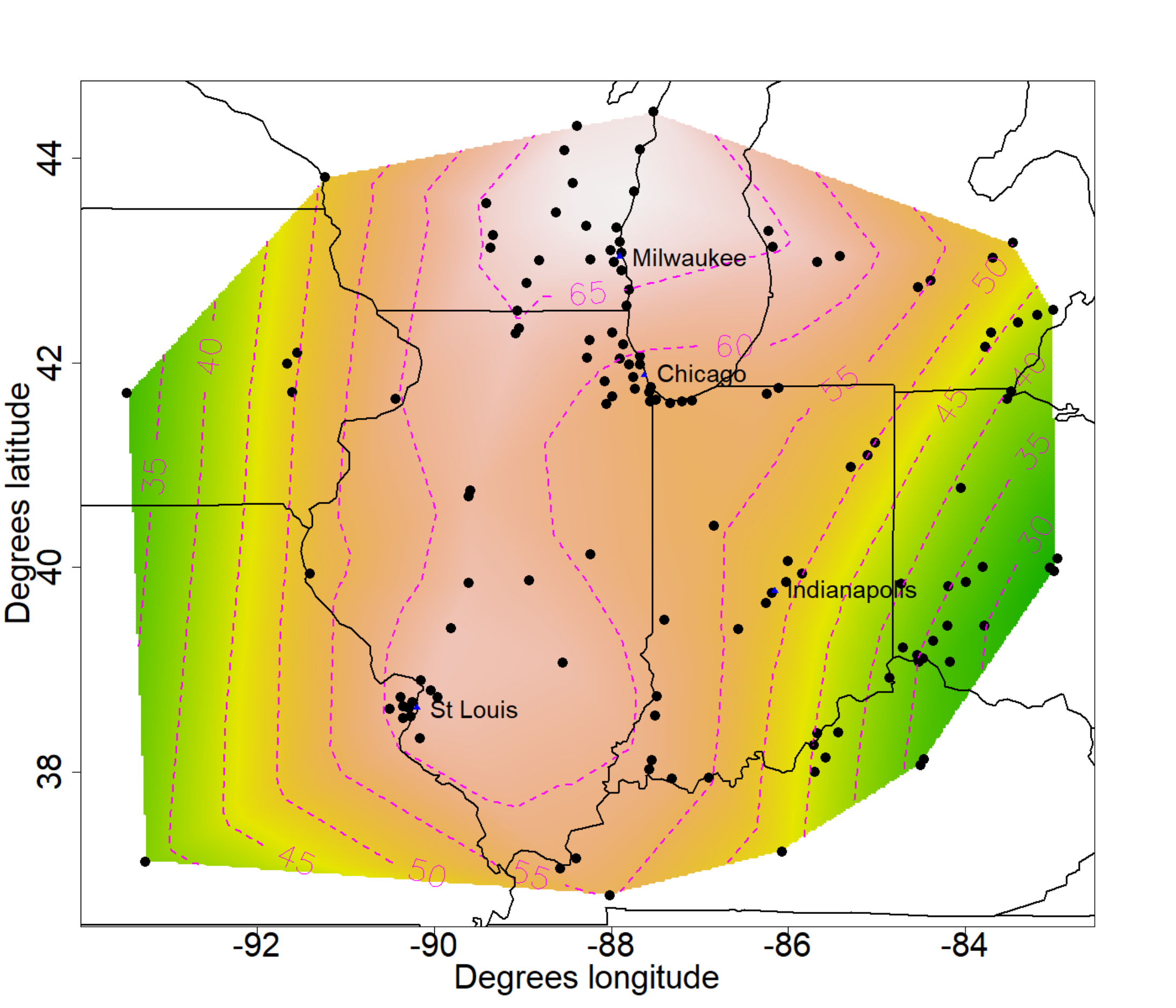"}} 
\caption{Contours of the estimated surfaces after removing outlying observations from the data. Left: contours of the least squares thin-plate spline estimator. Right: contours of the least absolute deviations thin-plate spline estimator. Darker colors indicate higher ozone concentrations.  The observation sites are depicted with solid black dots.}
\label{fig4}
\end{figure}

To demonstrate precisely how sensitive the LS estimator is towards outlying observations, we have removed the LAD-detected outliers from the dataset and recomputed the estimates. This results in  the left and right panels in Figure~\ref{fig4} for the LS and LAD estimates, respectively. It is interesting to observe that without these outliers the LAD and LS estimates are much closer to one another. Moreover, comparing the right panels of Figure~\ref{fig3} and Figure~\ref{fig4} reveals that, contrary to LS estimates, LAD estimates have undergone minimal change after removing the outliers from the data. LAD estimates are thus better able to describe the bulk of the data than their sensitive LS counterparts.

\section{Conclusion}

The present paper introduces robust estimators for multivariate nonparametric regression models based on the highly advantageous thin-plate penalty. The proposed class of estimators enjoys optimal theoretical properties under mild assumptions and can be expediently computed even in high dimensions. There are several research directions worth exploring from here. Our theoretical treatment relies on a specific rate of decay of $\lambda$ and it is not known whether $\lambda$  attains this rate of decay when determined by our robust cross validation procedure. In order to ascertain whether this is the case, a close investigation of robust model selection procedures can be worthwhile. Additionally, our treatment rests upon the convexity of the loss function and hence the important class of redescending thin-plate estimators is left out. Future effort can therefore be directed towards a dedicated treatment of these estimators. 

An important generalization of our ideas would involve robust estimation of multivariate nonparametric generalized linear models, where the distribution of the response variable can be any member of the exponential family, thus significantly extending the range of applications. Such an estimator may be based on, for example, the thin-plate penalty proposed herein and the density power divergence, as used by \citet{Cl:2022} in the univariate setting.

Another important area of research where thin-plate splines are likely to be successful is functional data analysis and, in particular, location and dispersion estimation from discretely sampled functional data. For the location case, thin-plate splines may be used to extend the robust estimator of \citet{Kalogridis:2022} for discretely sampled functional data on a bounded interval to discretely sampled functional data on much more complicated multivariate domains, such as a sphere. For dispersion estimation, thin-plate splines can be combined with resistant loss functions and provide a potent alternative to estimators based on tensor product penalties, see, e.g., \citep[Chapter 8]{Hsing:2015}. We aim to study these important generalizations as part of our future work.

\section*{Appendix: Proofs of the theoretical results}
\addcontentsline{toc}{section}{Appendix: Proofs of the theoretical results}

\begin{proof}[Proof of Proposition~\ref{prop:1}]
The result would follow by direct application of Theorem 3.2 of \citet{Cox:1985} provided that we can check the conditions of that theorem. Since $I_m^2$ is a squared semi-norm on $\mathcal{H}^{m}(\mathbbm{R}^d)$ and $I_m$ has a finite $M$-dimensional null space, these conditions entail the convexity and weak lower semicontinuity of the map 
\begin{align*}
\mathcal{H}^{m}(\mathbbm{R}^d) \to \mathbb{R}_{+} :  f \mapsto n^{-1} \sum_{i=1}^n \rho(Y_i - f(\mathbf{x}_i)).
\end{align*}
We only need to check weak lower continuity, because convexity follows easily from the convexity of $\rho$. Convexity also implies that we only need to establish the lower semicontinuity of the map, as by the Hahn-Banach theorem \citep{Rynne:2008} convexity and lower semicontinuity imply weak lower semicontinuity. To that end, let $\{f_k \}_k$ denote a sequence in $\mathcal{H}^{m}(\mathbbm{R}^d)$ converging to some $f^{\star}$. By the Sobolev embedding theorem \citep[Theorem 4.12]{Adams:2003}, we have
\begin{align*}
\max_{1 \leq i \leq n} |f_k(\mathbf{x}_i) - f^{\star}(\mathbf{x}_i)| \leq c_0 \| f_k - f^{\star} \|_{\mathcal{H}^{m}(B_{r_n}^d(0))} \leq c_0 \| f_k - f^{\star} \|_{\mathcal{H}^{m}(\mathbb{R}^d)},
\end{align*}
for some $c_0>0$ not depending on $k$, where $B_{r_n}^d(0)$ denotes the smallest ball in $\mathbbm{R}^d$ containing all the $\mathbf{x}_i$ and $\| \cdot \|_{\mathcal{H}^{m}(B_{r_n}^d(0))},\| \cdot \|_{\mathcal{H}^{m}(\mathbb{R}^d)} $ denote the standard Sobolev norms on $\mathcal{B}_{r_n}^d(0)$ and $\mathbbm{R}^d$, respectively. Letting $k \to \infty$, we find that $\max_{1 \leq i \leq n} |f_k(\mathbf{x}_i) - f^{\star}(\mathbf{x}_i)| \to 0 $. The continuity of $\rho$ concludes the proof.
\end{proof}

We now introduce some useful notation that will be used in the proof of Theorem~\ref{thm:1} and Theorem~\ref{thm:2}. Let $Q_n$ denote the probability measure on $(\mathbb{R}^d, \mathcal{B}(\mathbbm{R}^d))$ given by
\begin{align*}
Q_n(A) = \frac{1}{n} \sum_{i=1}^n \mathcal{I}(\mathbf{x}_i \in A), \ A \in \mathcal{B}(\mathbbm{R}^d),
\end{align*}
where $\mathcal{I}(\cdot)$ denotes the indicator function. Let $\mathcal{F}$ denote a class of real-valued functions on some domain $\mathcal{X} \subset \mathbb{R}^d$. The $\delta$-entropy for $\mathcal{F}$ in the $\mathcal{L}^2(Q_n)$-norm, $H(\delta, \mathcal{F}, \mathcal{L}^2(Q_n))$, is defined as the logarithm of the smallest number $N$ for which there exists a collection of functions $f_1, \ldots, f_N$ such that for every $f \in \mathcal{F}$ there exists a $j = j(f) \in \{1, \ldots, N\}$ with the property
\begin{align*}
\left\{\int_{\mathcal{X}} |f(\mathbf{x}) - f_j(\mathbf{x})|^2 d Q_n(\mathbf{x})  \right\}^{1/2} \leq \delta.
\end{align*}
Similarly, we define the $\delta$-entropy with respect to the supremum norm, $H_{\infty}(\delta, \mathcal{F})$, as the logarithm of the smallest $N$ for which there exists $f_1, \ldots, f_N$ such that for every $f \in \mathcal{F}$ there is a $j = j(f)$ with the property
\begin{align*}
\sup_{\mathbf{x} \in \mathcal{X}} |f(\mathbf{x}) - f_j(\mathbf{x})| \leq \delta.
\end{align*}
It is clear that $H(\delta, \mathcal{F}, \mathcal{L}^2(Q_n)) \leq H_{\infty}(\delta, \mathcal{F})$ for all $\delta>0$. To lighten the notation in all our proofs below we will use $c_0$ to denote generic constants. Thus, the value of $c_0$ may change from appearance to appearance.

\begin{lemma}
\label{lem:1}
Assume A3 and A4 and that $2m>d$. Then, there exists a universal constant $c_0$, independent of $n$, such that
\begin{align*}
H(\delta, \{ f \in \mathcal{H}^{m}(\mathbb{R}^d),\|f \|_{\mathcal{L}^2(Q_n)} \leq 1, I_m(f) \leq M \}, \mathcal{L}^2(Q_n)  \} \leq c_0 \left( \frac{M}{\delta} \right)^{d/m}, \quad \delta>0, \ M\geq 1.
\end{align*}
\end{lemma}
\begin{proof}
Observe that by A3, $Q_n$ is restricted to the open set $\mathcal{O}$ and for any $A \in \mathcal{B}(\mathbbm{R}^d)$ we have $Q_n(A) = Q_n(A \cap \mathcal{O})$. Therefore, for any $(f,g) \in \mathcal{H}^{m}(\mathbb{R}^d) \times \mathcal{H}^{m}(\mathbb{R}^d) $ we find
\begin{align*}
\int_{\mathbb{R}^d}|f(\mathbf{x})-g(\mathbf{x})|^2 dQ_n(\mathbf{x}) & = \int_{\mathcal{O}}|f(\mathbf{x})-g(\mathbf{x})|^2 dQ_n(\mathbf{x})
\\  & \leq \sup_{\mathbf{x} \in \mathcal{O}}|f(\mathbf{x}) - g(\mathbf{x})|^2.
\end{align*}
Hence, to bound $H(\delta, \mathcal{F}, \mathcal{L}^2(Q_n)  \}$ it suffices to obtain a bound on
\begin{align}
\label{eq:A1}
H_{\infty}(\delta,\{ f \in \mathcal{H}^{m}(\mathcal{O}),\|f \|_{\mathcal{L}^2(Q_n)} \leq 1, I_m(f) \leq M \}), \quad \delta>0, \ M\geq 1.
\end{align}

To bound \eqref{eq:A1} we start by showing the following inclusion
\begin{align*}
\{ f \in \mathcal{H}^{m}(\mathcal{O}),\|f \|_{\mathcal{L}^2(Q_n)} \leq 1, I_m(f) \leq M \}  \subset \{ f \in \mathcal{H}^{m}(\mathcal{O}): \|f\|_{\mathcal{H}^{m}(\mathcal{O})}\leq c_0 M\}
\end{align*}
for some $c_0>0$ and all $M \geq 1$, where $\| \cdot \|_{\mathcal{H}^{m}(\mathcal{O})}$ is the Sobolev norm given by
\begin{align*}
\|f\|_{\mathcal{H}^{m}(\mathcal{O})} = \left\{ \int_{\mathcal{O}} |f(\mathbf{x})|^2 d \mathbf{x} + \sum_{m_1 + \ldots + m_d = m} \binom{m}{m_1, \ldots, m_d} \int_{\mathcal{O}} \left| \frac{\partial^m f(\mathbf{x})}{\partial x_{1}^{m_1} \ldots \partial x_{d}^{m_d}}  \right|^2 d \mathbf{x}  \right\}^{1/2}.
\end{align*}
To establish this inclusion, take an $f \in \mathcal{H}^{m}(\mathcal{O})$ such that $\| f \|_{\mathcal{L}^2(Q_n)} \leq 1$ and $I_m(f) \leq M$ for some $M \geq 1$. Our assumptions imply those of Theorem 3.4 in \citet{Utr:1988}, hence an application of that theorem reveals the existence of a constant $c_0$ such that
\begin{align*}
\int_{\mathcal{O}} |f(\mathbf{x})|^2 d \mathbf{x} & \leq c_0 \int_{\mathcal{O}} |f(\mathbf{x})|^2 d Q_n(\mathbf{x}) \\ & \quad +  c_0\sum_{m_1 + \ldots + m_d = m} \binom{m}{m_1, \ldots, m_d} \int_{\mathcal{O}} \left| \frac{\partial^m f(\mathbf{x})}{\partial x_{1}^{m_1} \ldots \partial x_{d}^{m_d}}  \right|^2 d \mathbf{x} 
\\ & \leq c_0 \left\{ 1  + I_m^2(f) \right\}
\\ & \leq c_0 M^2,
\end{align*}
where the last inequality follows from our assumption that $M \geq 1$. With this bound we now obtain
\begin{align*}
\|f\|_{\mathcal{H}^{m}(\mathcal{O})} \leq \left\{c_0 M^2 + I_m^2(f) \right\}^{1/2}  \leq c_0 M,
\end{align*}
as claimed. 

The final step of the proof is a bound on 
\begin{align*}
H_{\infty}\left(\delta, \{f \in \mathcal{H}^{m}(\mathcal{O}): \|f\|_{\mathcal{H}^{m}(\mathcal{O})}\leq c_0 M \} \right), \ \delta>0, M \geq 1.
\end{align*}
But this is the entropy of the closed $c_0 M$-ball and Proposition 6 of \citet{Cucker:2001} implies the existence of a universal $c_0$ such that
\begin{align*}
H_{\infty}\left(\delta, \{f \in \mathcal{H}^{m}(\mathcal{O}): \|f\|_{\mathcal{H}^{m}(\mathcal{O})}\leq c_0 M \} \right) \leq c_0 \left(\frac{M}{\delta} \right)^{d/m}, \ \delta>0, M \geq 1.
\end{align*}
The proof is complete.

\end{proof}

\begin{proof}[Proof of Theorem~\ref{thm:1}] 

The proof of the theorem employs the convexity argument of \citet{van de Geer:2002} combined with the Sobolev embedding theorem in order to localize the behaviour of the objective function around $f_0$. A tight bound on the asymptotic variance is established with the help of an exponential inequality based on the improved entropy estimates obtained in Lemma~\ref{lem:1}.

Write $L_n(f) = M_n(f) + \lambda I_m^2(f)$ where $M_n(f) = n^{-1} \sum_{i=1}^n \rho(Y_i-f(\mathbf{x}_i))$. Observe that $L_n(f)$ is the sum of two convex functions and as such it is itself convex. By Proposition~\ref{prop:1} there exists a minimizer of $L_n(f)$ in $\mathcal{H}^{m}(\mathbb{R}^d)$, which we denote with $\widehat{f}_n$. Put
\begin{align*}
\widetilde{f}_n = \alpha \widehat{f}_n + (1-\alpha)f_0,
\end{align*}
for some $\alpha \in (0,1)$ to be chosen. As $f_0 \in \mathcal{H}^{m}(\mathbb{R}^d)$ we have
\begin{align*}
L_n(\widetilde{f}_n) \leq \alpha L_n(\widehat{f}_n) + (1-\alpha)L_n(f_0) \leq L_n(f_0),
\end{align*}
from where, after adding $\mathbb{E}\{M_n(\widetilde{f}_n) -M_n(f_0) \} $ on both sides, we get
\begin{align}
\label{eq:A2}
\mathbb{E}\{M_n(\widetilde{f}_n) -M_n(f_0) \}  +  I_m^2(\widetilde{f}_n) \leq \left[ M_n(f_0) - \mathbb{E}\{M_n(f_0)\} -  M_n(\widetilde{f}_n) + \mathbb{E}\{M_n(\widetilde{f}_n) \}  \right] + I_m^2(f_0).
\end{align}
We choose $\alpha = 1/(1+\|\widehat{f}_n-f_0 \|_{\mathcal{L}^2(Q_n)})$. Clearly, $\alpha \in (0,1)$ and 
\begin{align*}
\|\widetilde{f}_n - f_0 \|_{\mathcal{L}^2(Q_n)}  = \alpha\|\widehat{f}_n - f_0 \|_{\mathcal{L}^2(Q_n)} \leq 1.
\end{align*} 
Our proof consists of deriving a lower bound on the left-hand side of \eqref{eq:A2} and an upper bound on the right-hand side of \eqref{eq:A2}, both in terms of $\|\widetilde{f}_n-f_0\|_{\mathcal{L}^2(Q_n)}$. 

We begin with the lower bound. By assumption $2m>d$ and $\mathcal{O}$ is a bounded set satisfying the uniform cone condition. Therefore, by Sobolev's theorem \citep[Theorem 4.12]{Adams:2003} we have the (compact) embedding
\begin{align*}
\mathcal{H}^{m}(\mathcal{O}) \to \mathcal{C}(\mathcal{O}),
\end{align*}
which implies the existence of a universal $c_0>0$ such that for any $f \in \mathcal{H}^{m}(\mathcal{O}) \subset \mathcal{C}(\mathcal{O})$,
\begin{align*}
\sup_{\mathbf{x} \in \mathcal{O}} |f(\mathbf{x})| \leq c_0 \left\{ \int_{\mathcal{O}} |f(\mathbf{x})|^2 d \mathbf{x} +   I_m^2(f) \right\}^{1/2}.
\end{align*}
Approximating $\int_{\mathcal{O}} |f(\mathbf{x})|^2 d \mathbf{x}$ with $\int_{\mathcal{O}} |f(\mathbf{x})|^2 d Q_n(\mathbf{x})$, as in the proof of Lemma~\ref{lem:1}, yields
\begin{align*}
\sup_{\mathbf{x} \in \mathcal{O}} |f(\mathbf{x})| \leq c_0 \left\{ \int_{\mathcal{O}} |f(\mathbf{x})|^2 d Q_n(\mathbf{x}) +   I_m^2(f) \right\}^{1/2}.
\end{align*}
Now, since $I_m^2$ is a squared semi-norm and $I_m^2(f_0)$ is bounded, $I_m^2(f_0) \leq 1$, say, this inequality reveals that for all $f \in \mathcal{H}^m(\mathbbm{R}^d)$ satisfying $\|f-f_0\|_{\mathcal{L}^2(Q_n)} \leq 1$ we have
\begin{align*}
\sup_{\mathbf{x} \in\mathcal{O}}| f(\mathbf{x})-f_0(\mathbf{x})|\leq c_0 \left\{ 1 + I_m^2(f - f_0) \right\}^{1/2} \leq c_0 \{1+ I_m^2(f)\}^{1/2}.
\end{align*}
It follows that we can choose a large enough $D_{\kappa}>1$, not depending on $f$, such that 
\begin{align*}
\sup_{\mathbf{x} \in \mathcal{O}} \frac{| f(\mathbf{x})-f_0(\mathbf{x})|}{D_{\kappa}\{1+ I_m^2(f)\}^{1/2}} \leq \kappa,
\end{align*}
where $\kappa$ is the constant in assumption A2. Therefore, for all $f \in \mathcal{H}^m(\mathbbm{R}^d)$ satisfying $\|f-f_0\|_{\mathcal{L}^2(Q_n)} \leq 1$, by A2, we find
\begin{align*}
\mathbb{E}\{M_n(f) -M_n(f_0) \} & = \frac{1}{n} \sum_{i=1}^n \mathbb{E} \left\{ \rho\left(\epsilon_i+f_0(\mathbf{x}_i)-f(\mathbf{x}_i) \right) - \rho\left(\epsilon_i \right) \right\}
\\ &\geq  \frac{1}{n} \sum_{i=1}^n \mathbb{E} \left\{ \rho\left( \epsilon_i + \frac{f_0(\mathbf{x}_i)-f(\mathbf{x}_i)}{D_k\left\{1+I_m^2(f)\right\}^{1/2}}  \right) - \rho\left(\epsilon_i \right) \right\}
\\ & \geq \kappa \frac{\left\|f-f_0\right\|^2_{\mathcal{L}^2(Q_n)}}{D_k^2\left\{1+I_m^2(f) \right\} }
\end{align*}
It follows that 
\begin{align*}
\inf_{f \in \mathcal{H}^m(\mathbbm{R}^d): \|f-f_0\|_{\mathcal{L}^2(Q_n)} \leq 1} \left[ \frac{\mathbb{E}\{M_n(f) -M_n(f_0) \}}{\kappa \frac{\|f-f_0\|^2_{\mathcal{L}^2(Q_n)}}{D_{\kappa}^2 \{1+I_m^2(f) \}}} \right] \geq 1
\end{align*}
and from this, since, by construction, $\widetilde{f}_n \in \mathcal{H}^m(\mathbbm{R}^d)$ and $\|\widetilde{f}_n - f_0\|_{\mathcal{L}^2(Q_n)} \leq 1$, we see that 
\begin{align}
\label{eq:A3}
\mathbb{E}\{M_n(\widetilde{f}_n) -M_n(f_0) \} & = \frac{\mathbb{E}\{M_n(\widetilde{f}_n) -M_n(f_0) \}}{ \kappa \frac{\|\widetilde{f}_n-f_0\|^2_{\mathcal{L}^2(Q_n)}}{D_{\kappa}^2 \{1+I_m^2(\widetilde{f}_n) \}}} \kappa \frac{\|\widetilde{f}_n-f_0\|^2_{\mathcal{L}^2(Q_n)}}{D_{\kappa}^2 \{1+I_m^2(\widetilde{f}_n) \}} \nonumber
\\ & \geq \kappa \frac{\|\widetilde{f}_n-f_0\|^2_{\mathcal{L}^2(Q_n)}}{D_{\kappa}^2 \{1+I_m^2(\widetilde{f}_n) \}} \inf_{f \in \mathcal{H}^m(\mathbbm{R}^d): \|f-f_0\|_{\mathcal{L}^2(Q_n)} \leq 1} \left[ \frac{\mathbb{E}\{M_n(f) -M_n(f_0) \}}{\kappa \frac{\|f-f_0\|^2_{\mathcal{L}^2(Q_n)}}{D_{\kappa}^2 \{1+I_m^2(f) \}}} \right] \nonumber
\\ & \geq \kappa \frac{\|\widetilde{f}_n-f_0\|^2_{\mathcal{L}^2(Q_n)}}{D_{\kappa}^2 \{1+I_m^2(\widetilde{f}_n) \}}
\end{align}
This provides a lower bound for the left-hand side of \eqref{eq:A2} and completes the first part of our derivation.

Next, we derive an upper bound for the right-hand side of \eqref{eq:A2}. To accomplish this, we need to derive the modulus of continuity of the mean-centered process $M_n(f) - \mathbb{E}\{M_n(f)\}$. We will apply Lemma 8.5 of \citet{van de Geer:2000} to this process. First, observe that by A1 for all $(f, g) \in \mathcal{H}^{m}(\mathbbm{R}^d) \times \mathcal{H}^{m}(\mathbbm{R}^d)$ we have
\begin{align*}
| \rho(Y_i - f(\mathbf{x}_i)) -\rho(Y_i - g(\mathbf{x}_i))| \leq c_0 |f(\mathbf{x}_i) - g(\mathbf{x}_i)|,
\end{align*}
so that the lemma is applicable with $d_i(f,g) = |f(\mathbf{x}_i)-g(\mathbf{x}_i)|$ in the notation of \citet{van de Geer:2000}. In combination with Lemma~\ref{lem:1} we thus have
\begin{align*}
\sup_{f \in \mathcal{H}^m(\mathbbm{R}^d): \|f-f_0\|_{\mathcal{L}^2(Q_n)} \leq 1} \left| \frac{M_n(f_0) - M_n(f) - \mathbb{E}\{M_n(f_0) -  M_n(f) \}}{ n^{-1/2} \|f-f_0\|^{1-d/2m}_{\mathcal{L}^2(Q_n)}\{1+ I_m(f)\}^{d/2m}} \right| = O_{P}(1),
\end{align*}	
so that
\begin{align}
\label{eq:A4}
M_n(f_0) - M_n(\widetilde{f}_n) - \mathbb{E}\{M_n(f_0) -  M_n(\widetilde{f}_n) \} = O_P(n^{-1/2}) \|\widetilde{f}_n-f_0\|^{1-d/2m}_{\mathcal{L}^2(Q_n)}\{1+ I_m(\widetilde{f}_n)\}^{d/2m},
\end{align}
which provides the desired upper bound for the right-hand side of \eqref{eq:A2}.

Plugging the lower bound in \eqref{eq:A3} and the upper bound in \eqref{eq:A4} into \eqref{eq:A2}, we finally obtain
\begin{align}
\label{eq:A5}
c_0\frac{\|\widetilde{f}_n - f_0 \|^2_{\mathcal{L}^2(Q_n)}}{1+I_m^2(\widetilde{f}_n) } + \lambda I_m^2(\widetilde{f}_n) \leq O_P(n^{-1/2})\|\widetilde{f}_n-f_0\|^{1-d/2m}_{\mathcal{L}^2(Q_n)}\{1+ I_m(\widetilde{f}_n)\}^{d/2m} +  \lambda I_m^2(f_0).
\end{align} 
This inequality implies that
\begin{align}
\label{eq:A6}
\|\widetilde{f}_n-f_0\|^2_{\mathcal{L}^2(Q_n)}= O_P(n^{-2m/(2m+d)}) \quad \text{and} \quad I_m(\widetilde{f}_n) = O_P(1).
\end{align}
To see this implication, note that if for real numbers $a,b,c$ we have $a \leq b +c $ then either $a \leq 2b$ or $a \leq 2c$. Applying this on \eqref{eq:A5} leads to either
\begin{align}
\label{eq:A7}
c_0\frac{\|\widetilde{f}_n - f_0 \|^2_{\mathcal{L}^2(Q_n)}}{1+I_m^2(\widetilde{f}_n) } + \lambda I_m^2(\widetilde{f}_n) \leq O_P(n^{-1/2}) \|\widetilde{f}_n-f_0\|^{1-d/2m}_{\mathcal{L}^2(Q_n)}\{1+ I_m(\widetilde{f}_n)\}^{d/2m},
\end{align}
or 
\begin{align}
\label{eq:A8}
c_0\frac{\|\widetilde{f}_n - f_0 \|^2_{\mathcal{L}^2(Q_n)}}{1+I_m^2(\widetilde{f}_n) } + \lambda I_m^2(\widetilde{f}_n) \leq 2 \lambda I_m^2(f_0).
\end{align}
If \eqref{eq:A8} holds, \eqref{eq:A6} is easily verified. On the other hand, if \eqref{eq:A7} holds, solving it we get
\begin{align*}
\|\widetilde{f}_n - f_0 \|_{\mathcal{L}^2(Q_n)} = O_P(n^{-m/(2m+d)})\{1+I_m(\widetilde{f}_n)\}^{\frac{d}{2m+d}}\{1+I_m^2(\widetilde{f}_n)\}^{\frac{2m}{2m+d}},
\end{align*}	
as well as
\begin{align*}
\frac{I_m^2(\widetilde{f}_n)}{\left\{1+I_m(\widetilde{f}_n)\right\}^{\frac{2d}{2m+d}}\left\{1+I_m^2(\widetilde{f}_n) \right\}^{\frac{2m-d}{2m+d}}}= O_P(n^{-2m/(2m+d)}) \lambda^{-1}.
\end{align*}
By our assumptions on $\lambda$, the latter implies that $I_m(\widetilde{f}_n) = O_P(1)$. Hence, $\|\widetilde{f}_n - f_0 \|_{\mathcal{L}^2(Q_n)} = O_P(n^{-m/(2m+d)})$ verifying \eqref{eq:A6} again.

The last step in our proof involves passage from $\widetilde{f}_n$ to $\widehat{f}_n$.	For this, first note that by definition of the convex combination $\widetilde{f}_n$ and \eqref{eq:A6},
\begin{align*}
\|\widetilde{f}_n - f_0\|_{\mathcal{L}^2(Q_n)} = \frac{\|\widehat{f}_n - f_0\|_{\mathcal{L}^2(Q_n)}}{1+\|\widehat{f}_n - f_0\|_{\mathcal{L}^2(Q_n)}}= O_P(n^{-m/(2m+d)}),
\end{align*}
whence also $\|\widehat{f}_n - f_0\|_{\mathcal{L}^2(Q_n)} = O_P(n^{-m/(2m+d)})$. Finally, since again by \eqref{eq:A6} $I_m(\widetilde{f}_n)= O_P(1)$, by the triangle inequality we get
\begin{align*}
I_m(\alpha (\widehat{f}_n-f_0)) \leq I_m(\widetilde{f}_n) + I_m(f_0) = O_P(1).
\end{align*}
But then also $I_m( \widehat{f}_n-f_0)= O_P(1)$, which implies the result. The proof is complete.

\end{proof}

\begin{proof}[Proof of Corollary~\ref{cor:1}]

The first part of the Corollary follows from Theorem 3.4 of \citet{Utr:1988} which in our case reads 
\begin{align*}
\int_{\mathcal{O}}|\widehat{f}_n(\mathbf{x}) - f_0(\mathbf{x})|^2 d \mathbf{x} & \leq c_0 \int_{\mathcal{O}}|\widehat{f}_n(\mathbf{x}) - f_0(\mathbf{x})|^2 d Q_n(\mathbf{x}) \\ & \quad + c_0 h_{\max,n}^{2m} \sum_{m_1 + \ldots + m_d = m} \binom{m}{m_1, \ldots, m_d} \int_{\mathcal{O}} \left| \frac{\partial^m (\widehat{f}_n(\mathbf{x}) - f_0(\mathbf{x}))}{\partial x_{1}^{m_1} \ldots \partial x_{d}^{m_d}}  \right|^2 d \mathbf{x},
\end{align*}
for some constant $c_0$ that does not depend on either $\widehat{f}_n$ or $f_0$. Applying now Theorem~\ref{thm:1} and our assumption $h_{\max,n} = O(n^{-1/(2m+d)})$ yields
\begin{align*}
\int_{\mathcal{O}}|\widehat{f}_n(\mathbf{x}) - f_0(\mathbf{x})|^2 d \mathbf{x} & = O_P(n^{-2m/(2m+d)}) + c_0 h_{\max,n}^{2m} O_P(1)
\\ & = O_P(n^{-2m/(2m+d)}) + O_P(n^{-2m/(2m+d)})
\\ & = O_P(n^{-2m/(2m+d)}),
\end{align*}
as asserted.

For the second part of the Corollary we apply an interpolation inequality due to \citet{Nir:1959}, a simplified version of which states
\begin{align*}
\int_{\mathcal{O}} \left|\frac{\partial f^j(\mathbf{x})}{\partial x_1^{j_1} \ldots \partial x_d^{j_d} } \right|^2 d \mathbf{x} & \leq c_0\left\{ \int_{\mathcal{O}} \left|\frac{\partial f^m(\mathbf{x})}{\partial x_1^{m_1} \ldots \partial x_d^{m_d} } \right|^2 d \mathbf{x} \right\}^{j/m} 
 \left\{ \int_{\mathcal{O}} |f(\mathbf{x}) |^2 d \mathbf{x} \right\}^{1-j/m}
\\& \quad  +  c_0\int_{\mathcal{O}} |f(\mathbf{x}) |^2 d \mathbf{x},
\end{align*}
for every $f\in \mathcal{H}^{m}(\mathcal{O})$, where $c_0$ is a universal constant and the inequality holds for all tuples $(j_1, \ldots, j_d)$ and $(m_1, \ldots, m_d)$ such that $j_1+\ldots+ j_d = j$ and $m_1+\ldots+m_d = m$, respectively. Now apply this inequality with $f$ replaced by $\widehat{f}_n - f_0$ and use the first part of the corollary and Theorem~\ref{thm:1} to get
\begin{align*}
\int_{\mathcal{O}} \left| \frac{\partial \widehat{f}_n^j(\mathbf{x})}{\partial x_1^{j_1} \ldots \partial x_d^{j_d} } -  \frac{\partial f^j_0(\mathbf{x})}{\partial x_1^{j_1} \ldots \partial x_d^{j_d} } \right|^2 d \mathbf{x} & = O_P(1) O_P(n^{-2(m-j)/(2m+d)}) + O_P(n^{-2m/(2m+d)})
\\ & = O_P(n^{-2(m-j)/(2m+d)}),
\end{align*}
which completes the proof.
\end{proof}

\begin{proof}[Proof of Theorem~\ref{thm:2}]

For uniformly sub-Gaussian errors $\epsilon_i$ and under Lemma~\ref{lem:1}, the derivation on \citet[p. 168]{van de Geer:2000} shows that
\begin{align*}
\frac{1}{n} \sum_{i=1}^n \epsilon_i(\widehat{f}_n(\mathbf{x}_i) - f_0(\mathbf{x}_i)) = O_P(n^{-1/2})\|\widehat{f}_n- f_0 \|^{1-d/2m}_{\mathcal{L}^2(Q_n)} \{1+I_m(\widehat{f}_n) \}^{d/2m}.
\end{align*}
Plug this into \eqref{eq:4} to obtain
\begin{align*}
\|\widehat{f}_n - f_0 \|^2_{\mathcal{L}^2(Q_n)} + \lambda I_m^2(\widehat{f}_n) \leq O_P(n^{-1/2})\|\widehat{f}_n- f_0 \|^{1-d/2m}_{\mathcal{L}^2(Q_n)} \{1+I_m(\widehat{f}_n) \}^{d/2m} + \lambda I_m^2(f_0).
\end{align*}
Now argue as in the proof of Theorem~\ref{thm:1} to complete the proof.

\end{proof}

\section*{Acknowledgements}

The author is grateful to two anonymous referees, the associate editor and the editor for suggestions that lead to a much improved manuscript. He is also grateful to Stefan Van Aelst for helpful discussions. This research has been supported by the Research Foundation-Flanders (project 1221122N).

\end{document}